\documentclass[11pt]{llncs}

\usepackage{braket} 

\newcommand{\eq}[1]{(\ref{eq:#1})}

\newcommand{\thm}[1]{Theorem~\ref{thm:#1}}
\newcommand{\lem}[1]{Lemma~\ref{lem:#1}}

\newcommand{\poly}{\mathop{\mathrm{poly}}}
\newcommand{\mcn}{\mathop{\mathrm{mcn}}}
\newcommand{\ii}{{\mathrm{i}}}

\newcommand{\norm}[1]{\left\|{#1}\right\|}

\newcommand{\be}{\begin{equation}}
\newcommand{\ee}{\end{equation}}

\begin{document}
\title{Simulating sparse Hamiltonians \\ with star decompositions\thanks{Work supported by MITACS, NSERC, QuantumWorks, and the US ARO/DTO.}}
\author{Andrew M.\ Childs\inst{1,3} \and 
        Robin Kothari\inst{2,3}}

\institute{Department of Combinatorics \& Optimization, University of Waterloo \and
David R.\ Cheriton School of Computer Science, University of Waterloo \and
Institute for Quantum Computing, University of Waterloo}

\maketitle

\thispagestyle{plain}
\pagestyle{plain}


\begin{abstract}
We present an efficient algorithm for simulating the time evolution due to a sparse Hamiltonian. In terms of the maximum degree $d$ and dimension $N$ of the space on which the Hamiltonian $H$ acts for time $t$, this algorithm uses $(d^2(d+\log^* N)\norm{Ht})^{1+o(1)}$ queries.  This improves the complexity of the sparse Hamiltonian simulation algorithm of Berry, Ahokas, Cleve, and Sanders, which scales like $(d^4(\log^* N)\norm{Ht})^{1+o(1)}$. To achieve this, we decompose a general sparse Hamiltonian into a small sum of Hamiltonians whose graphs of non-zero entries have the property that every connected component is a star, and efficiently simulate each of these pieces.
\end{abstract}


\section{Introduction\label{sec:intro}}

Quantum simulation of Hamiltonian dynamics is a well-studied problem~\cite{Lloyd96,AT03,BAC+07} and is one of the main motivations for building a quantum computer. Since the best known classical algorithms for simulating quantum dynamics are inefficient, Feynman suggested that computers that are inherently quantum might be better at simulating quantum systems~\cite{Feynman82}. Besides simulating physics, Hamiltonian simulation has algorithmic applications, such as adiabatic optimization~\cite{FGG+00}, unstructured search~\cite{FG98}, and the implementation of continuous-time quantum walks~\cite{CCD+03,FGG08}.

The input to the Hamiltonian simulation problem is a Hamiltonian $H$ and a time $t$; the problem is to implement the unitary operator $e^{-\ii Ht}$ approximately. We say that a Hamiltonian acting on an $N$-dimensional quantum system can be simulated efficiently if there is a quantum circuit using $\poly(\log N,t,1/\epsilon)$ one- and two-qubit gates that approximates (with error at most $\epsilon$) the evolution according to $H$ for time $t$. Since the time evolution depends on the product $Ht$, the size of the circuit should also be bounded by a polynomial in some quantity measuring the size of $H$. When $H$ is sparse, most of its matrix norms have comparable values, so the complexity of simulating $H$ is not very sensitive to how its size is quantified.  It is conventional to require that the scaling be polynomial in $\norm{H}$, the spectral norm of $H$.

Lloyd presented a method for simulating quantum systems that can be described by a sum of local Hamiltonians~\cite{Lloyd96}. A Hamiltonian is called local if it acts non-trivially on at most a fixed number of qubits, independent of the size of the system.

This was later generalized by Aharonov and Ta-Shma~\cite{AT03} to the case of sparse (and efficiently row-computable) Hamiltonians. A Hamiltonian is sparse if it has at most poly($\log N$) nonzero entries in any row. It is efficiently row-computable if there is an efficient procedure to determine the location and matrix elements of the nonzero entries in each row. 

The complexity of this simulation was improved by Childs \cite{Childs04} and further improved by Berry, Ahokas, Cleve and Sanders~\cite{BAC+07}. Their algorithm has query complexity  $(d^4(\log^* N)\norm{Ht})^{1+o(1)}$, where $d$ is the maximum degree of the graph of the Hamiltonian $H$. These algorithms decompose the Hamiltonian into a sum of Hamiltonians, each of which is easy to simulate. In this paper, we present a different method of decomposing the Hamiltonian, giving an algorithm with query complexity $(d^2(d+\log^* N)\norm{Ht})^{1+o(1)}$.

Note that the simulation of Ref.~\cite{BAC+07} has also been improved using a completely different approach~\cite{Childs09,BC09}. That algorithm is more efficient in terms of all parameters except the error $\epsilon$, on which its dependence is considerably worse. The algorithm we present here maintains the same dependence on $\epsilon$ as in Ref.~\cite{BAC+07}, providing the best known method for high-precision simulation of sparse Hamiltonians.

\section{Hamiltonians and graphs}

A Hamiltonian $H$ acting on $n$ qubits is a $2^n\times 2^n$ Hermitian matrix. It can also be thought of as the weighted adjacency matrix of a graph on $2^n$ vertices, where the weights are complex numbers and the weight of an edge from $u$ to $v$ is the complex conjugate of the weight of the edge from $v$ to $u$. We call the undirected graph formed by connecting two vertices if and only if the edge between them has nonzero weight the \emph{graph of the Hamiltonian}. 

A Hamiltonian is said to be $d$-sparse if it has at most $d$ nonzero entries in each row (i.e., the maximum degree of its graph is $d$). We often associate properties of the graph of a Hamiltonian with the Hamiltonian itself. For instance, we might say ``$H$ is a forest,'' meaning that the graph of $H$ is a forest.

A \emph{star graph} is a tree in which one vertex (called the center) is connected to all the other vertices and there are no other edges. In other words, it is a complete bipartite graph $K_{1,r}$. We call a forest in which each tree is a star graph a \emph{galaxy}. 

A directed graph is a \emph{directed forest} (\emph{directed tree}) if its undirected graph is a forest (tree). A directed tree is an \emph{arborescence} if it has a unique root $v$ such that all edges point away from $v$. Alternately, there is exactly one directed path from $v$ to any other vertex $u$. In an arborescence, the edges are always directed from the parent to the child. A directed forest in which each tree is an arborescence is called a forest of arborescences.

We use several matrix norms in our analysis.  These include the spectral norm, $\norm{H}:=\max_{\norm{v}=1} \norm{Hv}$; the maximum entry norm, $\max(H):= \max_{ij} |H_{ij}|$; and the maximum column norm,  $\mcn(H):=\max_{j} \norm{He_j}$, where $e_j$ is the $j^{\mathrm{th}}$ column of the identity matrix.

\section{Problem description and previous results}

The problem is to approximately implement the unitary $e^{-\ii Ht}$ for a $d$-sparse and efficiently row-computable $N$-dimensional Hamiltonian $H$ for time $t$. As input, we are given black-box access to $H$, and the values of $d$, $t$, and $N$. Since the Hamiltonian is sparse and efficiently row-computable, there is a convenient black-box formulation of the problem that abstracts away the details of computing matrix entries and locations. The Hamiltonian is provided as a black-box function $f$, which accepts a row index and an integer $i \in \{1,2,\ldots,d\}$ and outputs the column index and matrix element corresponding to the $i^\mathrm{th}$ nonzero entry in that row, if one exists. More precisely, if the nonzero elements in row $x$ are $y_1,y_2,\ldots, y_{d_x}$, where $d_x \leq d$ is the degree of $x$, then $f(x,i) = (y_i, H_{x,y_i})$ for $i\leq d_x$ and $f(x,i)=(x,0)$ for $i > d_x$. This black box can be implemented efficiently if the Hamiltonian to be simulated is sparse and efficiently row-computable. 

For each row $x$, we allow the order in which the $y_i$ are given by the oracle to be arbitrary (but fixed). We do not assume that there is a convenient ordering, such as the increasing order of labels. To use the black box in a quantum circuit, we define an equivalent unitary matrix $U_f$ which performs the operation $U_f\ket{x,i,0}=\ket{x,i,f(x,i)}$.

Let us denote the minimum number of queries to $U_f$ required to approximately simulate $e^{-\ii Ht}$ (up to error $\epsilon$, as quantified by the trace distance) by $Q(H,t)$. A common approach to this problem breaks it into two subproblems, which we call the Hamiltonian decomposition problem and the Hamiltonian recombination problem. First the Hamiltonian is decomposed into a sum of easy-to-simulate Hamiltonians; then these Hamiltonians are simulated for short times in a specific manner so that the overall simulation is approximately the same as that of $H$.

Since we will also follow the decomposition--recombination strategy, we review this approach as applied in Ref.~\cite{BAC+07}. The given Hamiltonian $H$ is decomposed into a sum of $m$ Hamiltonians, $H=\sum_{j=1}^m H_j$. Let $Q(H_j)$ denote the number of queries required to simulate $H_j$ for time $t'$ given black-box access to $H$. In general, the number of queries required might depend on $t'$, but in the simulations used here $Q(H_j)$ is independent of $t'$. Note that $Q(H_j)$ includes the number of queries required to decompose $H$ into $H_j$ as well as to simulate $H_j$. In Ref.~\cite{BAC+07}, the Hamiltonians $H_j$ are 1-sparse, and their decomposition uses $O(\log^* N)$ queries\footnote{The function $\log^*$ is defined by $\log^* N = 0$ if $N \le 1$ and $\log^* N = 1+\log^* \log N$ if $N > 1$.} to a black box for $H$. Since a 1-sparse Hamiltonian can be simulated with 2 queries given an oracle for the 1-sparse Hamiltonian~\cite{CCD+03,Childs04}, $Q(H_j)=O(\log^* N)$. More precisely,

\begin{theorem}[Hamiltonian edge decomposition~\cite{BAC+07}]
\label{thm:BACSHD}
If $H$ is an $N\times N$ Hamiltonian with maximum degree $d$, then there exists a decomposition $H=\sum_{j=1}^m H_j$, where each $H_j$ is 1-sparse, such that $m=6d^2$ and each query to any $H_j$ can be simulated by making $Q(H_j)=O(\log^* N)$ queries to $H$. 
\end{theorem}

These Hamiltonians are then recombined using the Lie--Trotter formula, which expresses the time evolution due to $H$ as a product of time evolutions due to the $H_j$. The unitary $e^{-\ii Ht}$ is approximated by a product of exponentials $e^{-\ii H_jt'}$, such that the maximum error in the final state does not exceed $\epsilon$. We want to upper bound the number of exponentials required, $N_\mathrm{exp}$. Reference~\cite{BAC+07} proves the following.

\begin{theorem}[Hamiltonian recombination~\cite{BAC+07}]
\label{thm:BACSHR}
Let $k$ be any positive integer.
If $H=\sum_{j=1}^m H_j$ is a Hamiltonian to be simulated for time $t$ by a product of exponentials  $e^{-\ii H_jt'}$, and the permissible error (in terms of trace distance) is bounded by $\epsilon \leq 1 \leq 2m5^{k-1}\norm{H}t$, then the number of exponentials required, $N_\mathrm{exp}$, is bounded by
\be
N_{\mathrm{exp}} \leq 5^{2k}m^2\norm{H}t\left(\frac{m\norm{H}t}{\epsilon}\right)^{1/2k}.
\ee
\end{theorem}

Using the upper bound on the number of exponentials and the number of queries needed to simulate any exponential, the total number of queries needed to simulate the Hamiltonian $H$ satisfies $Q(H,t) \leq N_\mathrm{exp} \times \max_{j} Q(H_j)$. With $Q(H_j)=O(\log^* N)$ and $m=6d^2$, we get
\be
\label{eq:BACSHS}
Q(H,t) = O\left(5^{2k} d^4 (\log^*N) \norm{H} t \left(\frac{d^2\norm{H}t}{\epsilon}\right)^{1/2k}\right).
\ee

We see that $Q(H,t)$ is almost linear in $t$, which is almost optimal due to a no--fast-forwarding theorem~\cite{BAC+07}. However, the dependence on $d$ is not optimal.  In the present paper we improve the dependence on $d$ without affecting the other terms. The dependence on $d$ has been improved in other approaches, but only at the expense of a worse dependence on the error $\epsilon$~\cite{Childs09,BC09}.

We propose a new algorithm for solving the Hamiltonian decomposition problem. This strategy breaks up the Hamiltonian into only $m=6d$ parts, but increases $Q(H_j)$ to $O(d + \log^* N)$, improving the overall dependence on $d$ and $N$.


\section{Hamiltonian decomposition}
\label{sec:HD}

The Hamiltonian decomposition problem is the problem of decomposing a Hamiltonian $H$ into a sum of $m$ Hamiltonians $H_j$ such that given a label $1\le j \le m$ and a time $t'$, the unitary $e^{-\ii H_jt'}$ can be efficiently simulated. 

We solve this problem by decomposing the Hamiltonian into $m=6d$ galaxies. To achieve this, we first decompose the given graph into $d$ forests using the forest decomposition technique of Paneconesi and Rizzi~\cite{PR01}. The idea is to assign one of at most $d$ colors to each edge of the graph (not necessarily a proper edge coloring) such that the edges of any particular color form a forest. Not only is this decomposition possible, but it has some special properties that are required later in \lem{vertexcolor}.

\begin{lemma}[Forest decomposition]
\label{lem:forestdec}
For any Hamiltonian $H$ of maximum degree $d$, there exists a decomposition $H=\sum_{c=1}^d H_c$ and an assignment of directions to the edges such that each $H_c$ is a forest of arborescences. Furthermore, given a color $c$ and a vertex $v$, we can determine $v$'s parent in $H_c$ with one query (or determine that it is a root) and with $O(d)$ queries we can determine the list of edges in $H_c$ incident on $v$. 
\end{lemma}

\begin{proof}
We first describe a procedure that assigns a color $c$ to each edge. $H_c$ then consists of all edges colored $c$. To color the edges, every vertex proposes a color for each edge incident on it using the oracle in the following way: if $f(x,i)=(y,H_{x,y})$, then $x$ proposes color $i$ for the edge $xy$. Similarly, $y$ proposes a color for the edge $xy$. The edge is now colored using the proposal of the vertex with higher label (i.e., if $x>y$ then the edge $xy$ is colored with $x$'s proposal). This coloring uses $d$ colors, which is optimal up to constants since a $d$-sparse graph can have $dn/2$ edges, but forests have at most $n-1$ edges. 

Now we assign directions to the edges and show that each $H_c$ has no cycles, which shows that each $H_c$ is a directed forest. The edge $xy$ is directed from $x$ to $y$ if $x < y$. This choice of directions results in a directed acyclic graph, which has no directed cycles. To rule out non-directed cycles, we note that any such cycle must contain a vertex $v$ for which both the edges of the cycle point toward $v$. This means the label of $v$ is greater than that of its two neighbors. Thus the color of these edges was decided by $v$, which cannot happen since vertices propose different colors for different edges. 

To show that each tree in $H_c$ is an arborescence, we show that it has a unique root. Observe that a directed tree with more than one root must have a vertex with more than one parent. This again leads to the situation where a vertex has two incoming edges of the same color, which is not possible since these edges are colored by this vertex's proposal.

To show that the parent of a vertex can be determined with one query, note that if $p_v$ is the parent of vertex $v$ in $H_c$, then the edge from $p_v$ to $v$ must be directed toward $v$. Thus the color of this edge is decided by $v$. If this edge is in $H_c$, it is colored $c$. So if $v$ has a parent, it must be the $c^\mathrm{th}$ neighbor of $v$. With one query to the oracle, we can determine the $c^\mathrm{th}$ neighbor of $v$. If there is no such neighbor, this vertex has no parent and is a root in $H_c$. Otherwise the output contains the label of the parent.

Finally, we show how to determine the list of edges in $H_c$ incident on $x$ with $O(d)$ queries. First we query the oracle at most $d$ times to get the labels of all the neighbors of $x$. For a neighbor $y$ where $y<x$, the edge between $x$ and $y$ is colored by $c$ only if $y$ is $x$'s parent in $H_c$. Thus we can discard all edges $xy$ where $y<x$ but $y$ is not the parent of $x$. When $y>x$, an edge between $x$ and $y$ is colored $c$ only if $x$ is $y$'s parent in $H_c$, and it takes one query to verify this for each $y$. Thus with at most $d$ additional queries we can determine if all such edges are colored $c$.
\qed\end{proof}

This lemma shows how to decompose a Hamiltonian into directed forests. Let $T$ be the Hamiltonian of such a forest. We will decompose $T$ into a sum of 6 galaxies, $T=T_1+T_2+\cdots+T_6$. This is achieved by using an extension of the ``deterministic coin tossing'' protocol of Cole and Vishkin~\cite{CV86} by Goldberg, Plotkin and Shannon~\cite{GPS88}. Their protocol gives a proper vertex coloring of an arborescence using only 6 colors making $O(\log^* N)$ queries. Vertex coloring a directed forest of arborescences gives a galaxy decomposition of the forest, since all the edges that point to vertices of a particular color form a galaxy.

\begin{lemma}[Vertex coloring a forest]
\label{lem:vertexcolor}
If $T$ is a forest of arborescences, and the parent of a vertex can be determined with one query to an oracle for $T$, then there exists a proper vertex coloring of $T$ using 6 colors, such that  the color of any vertex can be determined by making $O(\log^* N)$ queries. 
\end{lemma}

\begin{proof}
We first describe the vertex-coloring procedure for the forest. A simple observation is that we already possess a vertex coloring of the forest: the labels of the vertices. This is a trivial proper vertex coloring using $N$ colors. Now we use a procedure that decreases the number of colors used by a logarithmic factor. Then we can run several rounds of this procedure to decrease the number of colors down to 6. Let $c_j(x)$ be the color assigned to vertex $x$ at the beginning of the $j^\mathrm{th}$ round of the procedure. At the beginning of the first round, we have $c_1(x) = x$.

Let $x$ be a vertex with parent $p_x$. Assume that we started with a proper vertex coloring at the beginning of round $j$. Since we have a proper coloring, $c_j(x) \neq c_j(p_x)$. Let $k$ be the index of the first bit at which $x$ and $p_x$ differ, and let $b$ be the value of the $k^\mathrm{th}$ bit of $x$. The new color for vertex $x$ is the concatenation of $k$ and $b$, denoted $(k,b)$. If $x$ is the root, we take $k=0$.  We claim that if each vertex performs this procedure, the result is a proper vertex coloring. 

For a contradiction, suppose there are two adjacent vertices that have been assigned the same color in round $j$. Without loss of generality, one of them is the parent of the other, so let them be $y$ and its parent $p_y$. Since we started with a proper coloring at the beginning of round $j$, $c_j(y) \neq c_j(p_y)$, but now $c_{j+1}(y) = c_{j+1}(p_y)$. Let $c_{j+1}(y)=(k,b)$ where by definition $k$ is the bit at which  $c_j(y)$ and $c_j(p_y)$ differ, and $b$ is the value of the $k^\mathrm{th}$ bit of  $c_j(y)$. Since $c_{j+1}(p_y)$ also equals $(k,b)$, the $k^\mathrm{th}$ bit of $c_j(p_y)$ is $b$. But $c_j(y)$ and $c_j(p_y)$ are supposed to differ at the $k^\mathrm{th}$ bit, so this is a contradiction.  Therefore the coloring procedure is valid.

It remains to show that if the colors of the vertices are updated in this way, we reduce the number of colors to $6$ in $O(\log^* N)$ rounds. If $L_j$ is the number of bits used to represent colors at the beginning of round $j$, then $L_{j+1} = \lceil\log(L_j)\rceil+1$. Initially, $L_1=\lceil \log N \rceil$. This recurrence relation can be solved to yield $L_j\leq 3$ when $j=\log^* N$~\cite{GPS88}. 
Further rounds cannot decrease $L_j$ below 3, since $L_{j+1}=L_{j}$ when  $L_j=3$. A length of 3 bits allows the use of only 8 colors. Now we run the procedure once more. Since there are 3 possible values for $k$, and 2 for $b$, there are at most 6 different colors. The total number of rounds is now $\log^* N+1$. 

To show that the color of a vertex can be determined with $O(\log^* N)$ queries, we note that the color of vertex $x$ at the end of the first round depends solely on $x$ and $p_x$. In general, the color of vertex $x$ at the end of $j$ rounds depends only on its first $j$ ancestors. To determine $x$'s color after $\log^* N+1$ rounds, we need the labels of its  $\log^* N+1$ ancestors, which can be found with $\log^* N+1$ queries, since the parent of a vertex can be found with one query. 
\qed\end{proof}

We have shown that a Hamiltonian can be decomposed into $d$ forests of arborescences, each of which can be vertex-colored with 6 colors. If we consider all the edges of one of the $d$ forests that point to a vertex of a particular color, this graph is a galaxy. So this decomposes the original Hamiltonian into $6d$ galaxies. For this particular decomposition of the Hamiltonian to be useful, we need to show that galaxies can be simulated easily. 

\begin{theorem}[Galaxy simulation]
\label{thm:galaxysim}
If $H_j$ is a Hamiltonian whose graph is a galaxy of maximum degree $d$, and the oracle can identify which vertices are centers of stars, then the unitary operator $e^{-\ii H_jt}$ can be simulated using $O(d)$ calls to an oracle for $H_j$.
\end{theorem}

\begin{proof}
The key idea is that given a vertex $v$, we can learn everything about the star to which $v$ belongs in $O(d)$ queries. If $v$ is the center of the star, the oracle identifies it as the center, so we can query all its neighbors to learn everything about the star with at most $d$ queries. If $v$ is not the center, we can determine the center, which is the only neighbor of $v$, with only one query, and then learn the rest of the star with at most $d$ queries.

Let $R(x)$ denote all the information about the star to which $x$ belongs: the label of the center, the labels of the other vertices in some fixed order, and the weights of all the edges. It is essential that $R(x)$ depend only the star and not the particular vertex $x$ chosen from the star, so that if $x$ and $y$ belong to the same star then $R(x)=R(y)$. Since we know that $R(x)$ can be computed with $O(d)$ queries, we can implement the unitary $U$ given by $U\ket{x,0} = \ket{x,R(x)}$ with $O(d)$ queries.

The Hamiltonian we are trying to simulate, $H_j$, is a galaxy. Thus, if $c$ is the center of a star, and its neighbors are $y_i$ with edge weights $w_i$, then $H_j\ket{c} = \sum_i w_i\ket{y_i}$. If $x$ is not the center of a star, and the edge between $x$ and the center $c$ has weight $w_x$, then $H_j\ket{x}= w_x\ket{c}$. Let $K$ be a Hamiltonian which is similar to $H_j$, but acts on the input state $\ket{x,R(x)}$ instead of $\ket{x}$. That is, $K\ket{c,R(c)} = \sum_i w_i\ket{y_i,R(y_i)}$ when $c$ is the center, and $K\ket{x,R(x)}= w_x\ket{c,R(c)}$ otherwise. Note that although the second register looks different, it is unaffected by $K$ since $R(x)$ depends only on the star and not the vertex. Combining $K$ with the unitary $U$ above, we see that $H_j = U^\dag KU$. In words, $U$ first computes all the information about the star in another register, $K$ performs the required Hamiltonian, and the $U^\dag$ uncomputes the second register, which was unaffected by $K$.

This simulation is efficient since $K$ can be simulated efficiently.  More importantly for our purposes, $K$ requires no queries to implement, since all the information about the star is already present in the second register. Thus the operation $H_j = U^\dag KU$ requires only as many queries as $U$ and $U^\dag$ require, which is $O(d)$.
\qed\end{proof}

Combining \lem{forestdec}, \lem{vertexcolor}, and \thm{galaxysim} gives our Hamiltonian decomposition theorem.

\begin{theorem}[Hamiltonian star decomposition]
\label{thm:HD}
There exists a decomposition $H=\sum_{j=1}^m H_j$, where each $H_j$ is a galaxy, such that $m=6d$ and each galaxy $H_j$ can be simulated for time $t'$ using $Q(H_j)=O(d + \log^* N)$ queries to an oracle for $H$.
\end{theorem}

\begin{proof}
From \lem{forestdec}, \lem{vertexcolor}, and \thm{galaxysim}, we know that the claimed decomposition is possible. It remains to show that any $H_j$ can be simulated for time $t'$ using $O(d + \log^* N)$ queries. 

To show this, let us implement $H_j$ on the basis state $\ket{x}$. If the implementation is correct on all basis states, it is correct for all input states by linearity. We are given $1 \le c \le d$ and $1 \le t \le 6$, which together form the index $j$. We want to simulate the galaxy formed by edges in $H_c$ directed toward vertices colored $t$ by the vertex coloring algorithm of \lem{vertexcolor}. 

From the proof of \thm{galaxysim}, it is clear that if we can compute $R(x)$, then we can implement $U$, and thereby simulate the desired Hamiltonian. $R(x)$ contains all the information about the star to which $x$ belongs. Using the result of \lem{forestdec}, we can determine the list of $x$'s neighbors in $H_c$ using $O(d)$ queries. By the result of \lem{vertexcolor}, with  $O(\log^* N)$ queries we can determine $x$'s color according to the vertex coloring algorithm. 

If $x$'s color is not $t$, then $x$ must be the center of a (possibly empty) star in $H_j$. The only edges in this star point toward vertices of color $t$, so we compute the colors of all the children of $x$ in $H_c$. These can be computed using only the labels of their $\log^* N+1$ nearest ancestors, which are all common ancestors. Thus we can compute the colors of all of $x$'s children using $O(\log^* N)$ queries in total. Now we know the star around $x$, and thus $R(x)$, using $O(d+\log^* N)$ queries.

If $x$'s color is $t$, then $x$'s parent is the center of star. The parent of $x$, $p_x$, can be determined with one query. Since $x$ and $p_x$ are in the same star, $R(x) = R(p_x)$.  Since $p_x$ is the center of a star, we can compute $R(p_x)$ as described above; thus we can also compute $R(x)$.

We have shown that for any $x$, we can compute $R(x)$ with $O(d+\log^* N)$ queries. Thus the unitary $U$ in the proof of \thm{galaxysim} can be simulated with $O(d+\log^* N)$ queries. By \thm{galaxysim}, this means we can implement $H_j$ with $O(d+\log^* N)$ queries, as claimed.
\qed\end{proof}

Now we can use our Hamiltonian decomposition theorem with the Hamiltonian recombination theorem (\thm{BACSHR}). Since we have $Q(H_j)=O(d+\log^* N)$ from \thm{HD} and $m=6d$ from \lem{forestdec} and \lem{vertexcolor}, we get our final result using $Q(H,t)\leq N_\mathrm{exp} \times \max_{j} Q(H_j)$:
\be
\label{eq:HS}
Q(H,t) = O\left(5^{2k} d^2 (d+\log^*N) \norm{H} t \left(\frac{d\norm{H}t}{\epsilon}\right)^{1/2k}\right).
\ee

When compared with the query complexity of \eq{BACSHS}, we see that this improves the scaling with $d$. Furthermore, when $d=\Omega(\log^* N)$, which is likely to be the case when $d$ is not constant, \eq{HS} has no $\log^*N$ term: the scaling (in terms of $d$ and $N$) is $(d^3)^{1+o(1)}$, as compared to \eq{BACSHS} which scales like $(d^4\log^* N)^{1+o(1)}$.

\section{Remarks and conclusion}

So far, we have measured the size of $H$ using the spectral norm $\norm{H}$.  However, if we express the simulation complexity in terms of a different norm, then both \thm{BACSHR} and equation \eq{HS} can be improved to give slightly better bounds.

In the proof of \thm{BACSHR}, $\norm{H}$ is used as a simple upper bound for $\max_j \norm{H_j}$. However, omitting this step gives a slightly stronger version of \thm{BACSHR} with $\norm{H}$ replaced by $\max_j \norm{H_j}$.  For a 1-sparse Hamiltonian, $\norm{H_j}=\max(H_j)\leq\max(H)$~\cite{CK09}, so $\norm{H}$ can be replaced by $\max(H)$ in \eq{BACSHS}.  However, this also leads to an improvement of \eq{HS}. When $H_j$ is a galaxy, $\norm{H_j} = \mcn(H_j)$~\cite{CK09}, and since $H_j$ is entry-wise upper bounded by $H$, $\mcn(H_j) \leq \mcn(H)$. Thus $\norm{H}$ can be replaced with $\mcn(H)$ in \eq{HS}.  To directly compare the two simulations, we can apply the bound $\mcn(H) \leq \sqrt{d}\max(H)$~\cite{CK09} to express both query complexities in terms of $\max(H)$.  In these terms, we still find that star decomposition improves over edge coloring: our algorithm uses at most $(d^{2.5}(d+\log^* N)\max(Ht))^{1+o(1)}$ queries, whereas the algorithm of Ref.~\cite{BAC+07} scales like $(d^{4}(\log^* N)\max(Ht))^{1+o(1)}$.

In conclusion, we have described a Hamiltonian decomposition technique that reduces the query complexity of simulating sparse Hamiltonians. By the degree-dependent lower bounds established in Ref.~\cite{CK09}, we know that query complexities scaling like $\norm{H}$ or $\sqrt{d}\max(H)$ cannot be achieved. It would be interesting to see if the Hamiltonian decomposition--recombination framework can be used to further reduce the dependence on $d$ while keeping a similar dependence on the error $\epsilon$, or to establish stronger limitations on the simulation of sparse Hamiltonians taking error dependence into account.


\bibliographystyle{splncs}

\end{document}